\documentclass{article}

\usepackage{graphicx} % Required for inserting images
\usepackage{hyperref}
\usepackage{amsmath}
\usepackage{amssymb}
\usepackage{cleveref}
\usepackage{amsthm}
\usepackage{tikz}
\usepackage{multicol}
\usepackage{stmaryrd}

\usepackage{algorithm}
\usepackage[noend]{algpseudocode}
\usepackage{varwidth}

\usetikzlibrary{trees, positioning}

\usepackage[rm={proportional,lining},sf={proportional,lining},tt={tabular,lining}]{cfr-lm}

\theoremstyle{definition}
\newtheorem{theorem}{Theorem}

\newtheorem{lemma}{Lemma}

\crefname{lemma}{Lemma}{Lemmas}
\crefname{section}{Section}{Sections}
\crefname{figure}{Figure}{Figures}

\title{Earliest query answering over streamed trees}
\author{
    Mateusz Gienieczko\,\textsuperscript{1}\!,
    Martín Muñoz\,\textsuperscript{2}\!,
    Filip Murlak\,\textsuperscript{3}\!,\\
    Charles Paperman\,\textsuperscript{4}
}

\begin{document}

\maketitle

\begin{center}
\small
\textsuperscript{1}Technical University of Munich, Germany\\
\textsuperscript{2}Univ. Artois, CNRS, UMR 8188, Centre de Recherche en Informatique de Lens (CRIL), F-62300, France\\
\textsuperscript{3}University of Warsaw, Poland\\
\textsuperscript{4}CRIStAL, Université de Lille, INRIA, France
\end{center}

\begin{abstract}
Streaming allows executing queries over massive JSON or XML documents whose size makes it infeasible to fully parse them into a tree. Earliest query answering is a radical approach to reducing latency and memory footprint. To minimize latency, a document node must be returned as soon as  the node is guaranteed to be an answer regardless of how the document ends. Similarly, to minimize memory footprint, a node must be discarded as soon as it cannot become an answer regardless of how the document ends. For simple queries that select nodes based on the path from the root, the decision for each node can be made on the spot, but practical languages such as XPath or JSONpath support filters, which allow selecting nodes based on information collected from various parts of the document, possibly further down the stream. This makes earliest query answering a challenging task, as candidate nodes must be kept in memory until it becomes clear that they can be safely returned or discarded. We show that this can be done for all unary queries expressible in monadic second order logic (MSO), while ensuring constant update time---provided that nodes are returned by passing a suitable iterator, rather than one by one. 
\end{abstract}

\section{Introduction} \label{sec:intro}
Tree-shaped data occurs in many settings, from XML and JSON documents to
program syntax trees, logs, database dumps, and web data~\cite{abiteboul2000data}. A large body of work,
originating especially from the XML era, has studied how to extract
information from such data~\cite{gottlob_2005, xpath_survey_2008, abiteboul_1997, gottlob_complexity_xpath_2003, melnik_dremel_2010, json_2017}. Despite this long history, efficient query
evaluation over trees remains a relevant problem. In high-throughput settings
the cost of parsing the input and materializing a full in-memory representation
of the tree can already dominate the cost of query evaluation itself~\cite{palkar_filter_2018}.

A standard approach is to first build an in-memory representation of
the input tree and then evaluate the query on this representation. This
approach is often too expensive---both in time and in memory---especially for
massive semi-structured files. To obtain high performance, parsing and query
evaluation should instead be integrated: the evaluator should inspect the input
as it is read, maintain only the information that may still influence future
answers, and avoid materializing the whole tree whenever possible. This
naturally leads to {\sbweight streaming algorithms for querying trees}.

Streaming evaluation is relevant for
fully streaming query engines, such as
rsonpath~\cite{gienieczko_supporting_2024} and
JSONSki~\cite{jiang_jsonski_2022}, but also for preprocessing pipelines that
filter or partially analyze documents before loading them into a database, as
in Sparser~\cite{palkar_filter_2018}, or in more elaborate storage mechanisms
such as JSONTiles~\cite{durner_json_2021}. In this setting,
complex navigation and filtering constructs may force the query engine to maintain
candidate nodes across long portions of the input stream. This phenomenon
already appears for filters in XPath~\cite{robie_xml_2017}, and also in the
recently proposed JSONPath standard~\cite{gossner_jsonpath_2024}.

In a streaming setting, two efficiency requirements are particularly important.
The first is {\sbweight memory}: the algorithm should have {\sbweight a parsimonious memory footprint},
storing only information that may still affect future answers. The second is
{\sbweight delay}: {\sbweight answers should be produced as soon as possible}. Once the status of a
node is determined by the prefix read so far, keeping this node as a pending
candidate increases memory consumption and introduces unnecessary latency in returning the result set. Conversely, the decision must be made only once enough information is available. Earliest query answering~\cite{gauwin_queries_2011, gauwin_09} captures this tension by requiring each node
to be accepted or rejected as soon as its status is forced by the current
stream prefix.

In this paper, we represent trees as nested words~\cite{nested_word}, or equivalently as streams
of opening and closing events. We study Boolean and unary queries expressed in
monadic second-order logic (MSO)~\cite{comon2008tree}. A Boolean query asks whether the whole tree
satisfies a property. A unary query selects positions, or nodes, of the
input tree that satisfy a formula with one free first-order variable,
and the output is the set of positions satisfying the formula. 

One cannot hope, however, for \emph{constant} memory in general. Even simple unary
queries may require the algorithm to remember linearly many candidate positions.
For example, consider the query that selects all nodes whose label is equal to
the label of the last leaf in document order. Until the last leaf is read, the
algorithm does not know which label is relevant; in the worst case, it may
therefore need to keep enough information to output all positions with that
eventual label. Thus a $\Omega(n)$ memory lower bound is unavoidable for some queries, with $n$ the size of the input word.
Our goal is therefore not to guarantee small memory for every query. Rather, we
aim to maintain the necessary information efficiently, while ensuring that no
candidate is kept longer than what is logically necessary. 

\paragraph{Our contribution} 
We provide {\sbweight a streaming
algorithm for unary queries definable in MSO that complies with earliest query answering by maintaining a specialized data structure with constant-time updates} in data
complexity, in the RAM model. The constant may depend on the query, but not on
the size of the input prefix read so far. 
The algorithm performs each step with constant delay, and it emits each position at
the earliest prefix at which its membership in the query result set is determined.
More precisely, a position is output (resp. rejected) as soon as all possible completions of the
current stream prefix agree that this position belongs (resp. does not belong) to the answer.
This notion is independent of syntactic elements such as the closing event of the
corresponding node: depending on the query, a position may become determined either
before its subtree is closed, exactly when it is closed, or after. %
The algorithm is therefore optimal regarding the information available in the
stream prefix: no correct streaming algorithm can output or reject a position earlier.

\paragraph{Structure of the paper}
After this, in Section~\ref{sec:preliminaries}, we present some preliminary notions in order to state the main result.
Then, in Section~\ref{sec:boolean}, we show a simplified algorithm that does earliest query answering for {\em boolean} formulas; namely, the algorithm only has to decide whether the forest is a model of the formula or not, and do this as early as possible. 
After this  (Section~\ref{sec:unary}), we present the main result of the paper as an extension of the algorithm from the previous section that works for unary formulas. This algorithm uses a certain data structure that we first describe as an interface; we show its inner workings in the following section (Section~\ref{sec:ds}).
The last two sections are a discussion about related works in the literature (Section~\ref{sec:related}), and then conclusions and future work (Section~\ref{sec:concl}).

%\section{preliminaries}

\section{Preliminaries and Problem statement} \label{sec:preliminaries}

\paragraph{Forests and trees}
Our data will be represented by {\em ordered unranked labeled forests} (from now on, just {\em forests}). They contain \emph{nodes} which are labeled by a function $\lambda$ mapping nodes to an alphabet $\Sigma$. We use the standard notions of {\em parent}, {\em child}, {\em sibling}, etc., and we assume that the roots in the forest are also siblings.

We will also compose forests by concatenation and by application. If $F$ and $G$ are forests, then the forest $H = F\cdot G$ is the disjoint union of the two: if $G$ is not empty, then the rightmost root of $F$ becomes the left sibling of the leftmost root of $G$, and if $G$ (resp. $F$) is empty, then $H = F$ (resp. $H = G$). If $a\in \Sigma$ and $F$ is a forest, then $H = a(F)$ is a tree: if $F$ is non empty, then $H$ is a tree with a root that is labeled $a$, and that has the roots of $F$ as its children, and if $F$ is empty, then $H$ is a single node labeled $a$. Observe that for every nonempty forest $F$ there are two unique (possibly empty) forests $G$ and $H$, and a symbol $a\in\Sigma$ such that $F = a(G)\cdot H$. 

{\bf Important: } For presentation purposes, the title of the paper and the introduction talk about data represented in trees. Nonetheless, for the rest of the paper the input will always be forest, and the word ``tree'' will exclusively refer to a forest with one root. 

\paragraph{Automata model} We work with automata over unranked forests, in the first-child next-sibling variant. We assume the automata are bottom-up and {\em right-to-left} deterministic\footnote{It might seem counterintutive that we work with a model in which computation progresses leftward and upward, as it is essentially the opposite of what happens when we read the term encoding of a forest, but it proves convenient later on.}. Formally, an automaton is a tuple $\mathcal{A} = (\Sigma, Q, \delta, {\sf init}, {\sf Final})$ where $Q$ is a set of states, $\delta: \Sigma \times Q\times Q \to Q$ is the transition function, ${\sf init}\in Q$ is the sole initial state, and ${\sf Final}\subseteq Q$ is the set of final states. 

The semantics are as follows. Fix a forest $F$: we say that $\mathcal{A}$ evaluates a node $v$ in $F$ to a state $r$ if it evaluates its leftmost child to $p$, its right sibling to $q$, and $\delta(\lambda(v), p, q) = r$.
If the leftmost child or the right sibling do not exist, they evaluate to ${\sf init}$ by default.
Alternatively, we can recursively define $\delta(F) = \delta(a, \delta(G), \delta(H))$ if $F = a(G)\cdot H$, and $\delta(F) = {\sf init}$ if $F = \varepsilon$.
We say that $\mathcal{A}$ evaluates $F$ to $q$ if it evaluates its leftmost root to $q$; and we say it accepts $F$ if $q\in {\sf Final}$.

We assume that all states of $\mathcal{A}$ are {\em reachable}, by which we mean that for each $q\in Q$ there exists a forest that evaluates to $q$.

We briefly note that every Boolean query expressible in the monadic second-order logic on trees (MSO) is expressible as a forest automaton, and vice versa~\cite{ThatcherW68}; more on this later.

\paragraph{Forest encodings}
Given a forest $F$, the {\em term encoding} of $F$, denoted by ${\sf enc}(F)$ is a string in $\Sigma\cup \{\langle\,, \rangle\}$ that explicitly lays out a sequence of concatenations and applications that constitutes $F$: if $F$ is empty, then ${\sf enc}(F)$ is the empty string $\varepsilon$; otherwise, $F = a(G)\cdot H$ for some unique $a\in\Sigma$, $G$, $H$, and we recursively define ${\sf enc}(F) := a\,\langle\,{\sf enc}(G)\,\rangle\, {\sf enc}(H)$.
We say that a string $w$ over $\Sigma\cup \{\langle\,, \rangle\}$ is {\em well-nested} if there exists a forest $F$ such that ${\sf enc}(F) = w$.

We adopt a tokenized view of encodings: we see ${\sf enc}(F)$ as a sequence  $t_1 \cdots t_n$ where each $t_i$ is either an \emph{opening tag} of the form $a\,\langle$ for some $a\in\Sigma$ or the \emph{closing tag} $\rangle$. 

We will also work with prefixes of these term encodings. We say that a forest $F$ {\em agrees with} a word $w$ if $w$ is a prefix of ${\sf enc}(F)$ that ends with a symbol from $\{\langle\,, \rangle\}$; if some $F$ agrees with $w$, then we might simply say that $w$ is a {\em prefix} without indicating $F$.
For a prefix $w = w_0 \, c_1 \,\langle\, w_1 \cdots \, c_d \,\langle\, w_d$ where $w_0,\ldots,w_d$ are well-nested words, we call $c_1\langle\,,\ldots,c_d\langle$ the {\em open tags of $w$}, and we will sometimes single out $c_d\langle$ as the {\em last open tag} of $w$.
The {\em depth} of a prefix $w$, written as  ${\sf depth}(w)$, is its number of open tags.
We sometimes talk about the nodes in a prefix $w$, by which we mean the nodes that are common to every forest that agrees with $w$; equivalently, they are the nodes in the forest encoded by the word $w' = w \,\rangle\cdots \,\rangle\,$, in which $w$ is followed by ${\sf depth}(w)$ copies of the closing tag.
%We use the notation $\widehat w$ for the node that corresponds to the last open tag of $w$.

In this work, algorithms will read an input forest $F$ represented as the string ${\sf enc}(F)=t_1 t_2 \cdots t_n$, tag by tag. If, in the current step, the algorithm has read $i$ tags, then we call the prefix $t_1 \cdots t_i$ the {\em current prefix}.

\paragraph{Model of computation}
We adopt the computational model of \emph{Random Access Machines} (RAM) with register length of size $\Theta(\log(n))$ where $n$ is the {\em final} input size, and arithmetical computations on a single register are assumed to cost constant time. This is one of the standard models for streaming algorithms~\cite{ganardi_sliding_2022, munoz_streaming_2024, larsen_lower_2015}, so we do not linger on the implications of this choice. 

We want to provide precise guarantees on the space that the algorithm uses at each point of the computation. To this end  we assume the RAM has a lazy garbage collection system that only frees unused registers upon demand. At each step, we identify some {\em root registers} whose information is directly accessible in the algorithm, and we consider that the space utilized by the algorithm is the number of registers that are reachable from these root registers.
This implies that the algorithm might sometimes dispose of a single root register, and this effectively frees multiple registers in memory all at once.

\paragraph{Unary MSO queries} The evaluation task that we are set out to solve is, given a unary query expressed as an MSO formula $\varphi(x)$ with a single free first-order variable $x$ and a forest $F$ over some fixed alphabet $\Sigma$, output every node $v$ in $F$ such that $F,v\models\varphi(x)$. We write $\llbracket \varphi(x) \rrbracket(F)$ for the set of all such nodes in $F$.

A unary MSO query can be equivalently represented using an automaton $\mathcal{A}$ over the \emph{extended alphabet} $\Sigma^{\sf out} := \Sigma \cup \{\check{a} \mid a \in \Sigma \}$. For a node $v$ in $F$ let $F[v]$ be the result of replacing the label $a$ of $v$ by $\check{a}$, and define a function $\llbracket \mathcal{A} \rrbracket$ by $\llbracket \mathcal{A} \rrbracket(F) = \{v\mid \mathcal{A} \text{ accepts } F[v]\}$. Each MSO formula $\varphi(x)$ can be translated to a forest automaton $\mathcal{A}$ such that  $\llbracket \varphi(x) \rrbracket = \llbracket \mathcal{A} \rrbracket$~\cite{munoz_streaming_2024, LohreyS26}. In what follows we distinguish between {\em unmarked forests}, which are forests over $\Sigma$, and {\em marked forests}, which are over $
\Sigma^{\sf out}$ and contain exactly one node labeled $\check{a}$ for some $a\in\Sigma$.

\paragraph{Earliest query answering}
Consider a query $q$, given via an MSO formula or an automaton. We are interested in enumerating $\llbracket q \rrbracket(F)$ while reading ${\sf enc}(F)$, but we want to return node $v$ as soon as $v$ is guaranteed to be an answer, and to discard a candidate node $v$ as soon as $v$ is guaranteed to never become an answer. Let us formalize this. Consider a prefix $w$. We call node $v$
\begin{itemize}
    \item a {\em definite answer} if $v\in \llbracket q\rrbracket(F)$ for every forest $F$ that agrees with $w$;
    \item a {\em definite nonanswer} if $v\notin \llbracket q\rrbracket(F)$ for every forest $F$ that agrees with $w$;
    \item  {\em indefinite} or {\em candidate} otherwise; that is, when $v$ might become an answer in some forests agreeing with $w$ and a nonanswer in others.
\end{itemize}
We write the set of all candidate nodes for $w$ as ${\sf Cands}(q, w)$.

We are interested in the following kind of algorithms. The algorithm reads a streamed forest and keeps a collection of nodes in memory: upon reading the opening tag of a node, the node is either added to the collection, or returned as an answer, or discarded. The collection is also updated: some nodes are returned as answers, some are discarded, and some are kept. We require that each node $v$ is returned (resp. discarded) after reading the shortest prefix $w$ such that $v$ is a definite answer (resp. definite nonanswer) for $w$. This complies with the classical notion of earliest query answering as it was presented in~\cite{gauwin_09}.

Next, we refine this setting by bringing it closer to the machine model, in order to facilitate precise complexity guarantees. 

\paragraph{Streaming enumeration}
A {\em streaming evaluation} algorithm is an algorithm that receives a query $q$, and a stream $\mathcal{S}$ that reveals the symbols from an input word $w$ one by one from beginning to end---we say that $\mathcal{S}$ {\em streams} $w$. We consider a special function ${\sf yield}[S]$ that upon each call returns the next symbol of $w$, and after the last symbol it returns a flag {\sf end}. At the end of the stream, the algorithm should produce the result from evaluating $q$ over $w$.

In this work, we use a more refined notion that we call a {\em streaming enumeration} algorithm. We split the evaluation into alternating {\em update} phases, where the algorithm preprocesses a data structure that contains the answers, and {\em enumeration} phases, in which the answers are printed one by one in a dedicated write-only space. A full run of the algorithm will start with the first call to ${\sf yield}[\mathcal{S}]$, which marks the beginning of the first update phase; then it runs the first enumeration phase to print any definite answers; and then calls ${\sf yield}[\mathcal{S}]$ again to continue with the second update phase, and so on. When ${\sf yield}[\mathcal{S}]$ returns {\sf end}, the algorithm runs its last update and enumeration phases, and ends. By this point, the entire answer set should have been printed without repeats. In the literature, this type of evaluation algorithm is sometimes called an {\em enumeration algorithm with early output}~\cite{RiverosJV23}.

We measure the resources required by this algorithm across three dimensions:
\begin{itemize}
    \item {\em Update time. } This is the maximum number of steps taken in any update phase, namely, between any call to ${\sf yield}[\mathcal{S}]$, and the next enumeration phase.
    \item {\em Enumeration delay. } The maximum number of steps, across all enumeration phases, between: (1) beginning an enumeration phase and printing its first output, (2) printing any two consecutive outputs and (3) printing the last output of the current enumeration phase and the moment the phase ends.
    \item {\em Instantaneous space. } We measure the space at every step of the computation in terms of the {\em active registers} in memory: these are the ones that hold all the accessible information, as specified in the \emph{Model of computation} paragraph above. The algorithm uses instantaneous space $s(i)$ if this is the maximum number of active registers at the end of  $i$-th enumeration phase for any $i\in [1,|w|+1]$.
\end{itemize}

After this preamble we can state our main result. Let $f(\varphi)$ be some computable value that only depends on $\varphi$.

\begin{theorem}[Main result]
    There is a streaming enumeration algorithm that receives an MSO formula $\varphi(x)$ and a stream $\mathcal{S}$ streaming $w = {\sf enc}(F)$ and produces $\llbracket\varphi(x)\rrbracket(F)$. Further, this algorithm:
    \begin{itemize}
        \item works with $f(\varphi)$ update time,
        \item has constant-delay enumeration,
        \item and uses $O(\,|{\sf Cands}(\varphi(x), w[1,i])| + f(\varphi)\cdot{\sf depth}(w[1,i])\,)$ instantaneous space.
    \end{itemize}
\end{theorem}

\section{Warmup: Boolean queries} \label{sec:boolean}

Suppose that our goal is to solve the problem of earliest query answering for \emph{boolean queries}. That is, we have a fixed automaton $\mathcal{A} = (\Sigma, Q, \delta, {\sf init}, {\sf Final})$, we read a streamed encoding of a forest over $\Sigma$, and our task is to determine if the encoded forest is accepted or rejected, and we have to answer \emph{yes} or \emph{no} as soon as the answer of the automaton is determined by the already processed prefix of the encoding. 
We note that this result is implicit in the work of Gauwin et al.~\cite{gauwin_09}; it is only presented here for illustration purposes.

\begin{theorem}
There is an algorithm that receives an MSO formula $\varphi$ and a stream $\mathcal{S}$ streaming $w = {\sf enc}(F)$ and answers whether $F\models \varphi$ as early as possible, with constant update time and $f(\varphi)\cdot{\sf depth}(w[1,i])$ instantaneous space.
\end{theorem}\label{thm:boolean-case}

Reading the term encoding of a forest $F$ corresponds to a DFS on the forest, with children of a given node processed left-to-right. 
Having read a prefix $w$ of the term encoding of forest $F$ we are in a node $\widehat w$ of $F$ that corresponds to the most recent opening tag that has not been closed yet. If no such tag exists, we are in the \emph{virtual root} of the forest, labeled with a special symbol $\$$ such that $\delta(\$,q,q') = q$. 
We shall refer to $\widehat w$ as the \emph{current node}. By the \emph{current path} we mean the path from the root of the tree containing the current node to the current node. The \emph{visited} part of the forest consists of nodes whose opening tag we have already read---namely, the nodes in $w$. All visited nodes can be split into three categories: we have the current node ($C$, for current), the visited proper descendants of the current node ($B$, for below), and the remaining visited nodes ($A$, for above). Nodes from $B$ form a subforest of $F$, and nodes from $A$ form a {\em context} with multiple holes. Indeed, each node in $A$ that lies on the current path has a hole as its rightmost child, directly to the right of the next node on the current path. The last node in $A$ on the current path is the parent of $C$; it has a hole instead of $C$. We refer to it as the \emph{current hole} (see~\autoref{fig:dec}).

\begin{figure}
    \centering
    \includegraphics[width=0.3\linewidth]{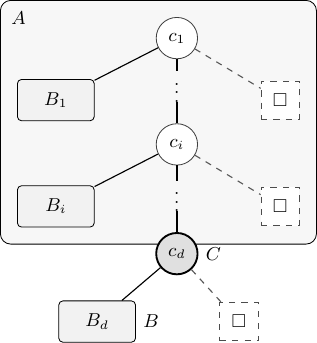}
    \caption{Decomposition of a tree into parts $A$, $B$ and $C$}
    \label{fig:dec}
\end{figure}

The nodes in the current path correspond to the open tags of $w$, which are labeled $c_1,\ldots,c_d = c$, and the already-seen subforests that descend from each are named $B_1,\ldots,B_d = B$. The union of all of these subforests, minus $C$ and $B$, corresponds to $A$. We can decompose  $A$ into a sequence of contexts: we call $A_0$ a context that is only a hole, and $A_i$ the context with $i+1$ holes that is obtained by plugging a node labelled $c_i$ with $B_i$ and a hole as children, and another hole as its right sibling.

To a forest $F$ we associate a function $\widehat F:Q\to Q$. We define it recursively by three cases: if $F$ is empty, then $\widehat F(p) = p$; if $F$ is the concatenation of two forests $F_1$ and $F_2$, then we define $\widehat F := \widehat F_1 \circ \widehat F_2$; if $F$ is a tree with a root labeled $c$ whose children form the forest $F_0$, we define $\widehat F(p) := \delta(c, \widehat F_0(\mathsf{init}), p)$. 
Intuitively, this function represents the final state that is reached once we run the automaton on $F$ after replacing its ``initial state from the right'' by $p$.
Note that this function is well defined given that the automaton is deterministic and trimmed.

We extend this definition to work with {\em $d$-ary contexts}, which extend forests by including $d$ nodes that are labeled with a special symbol $\square$. They are defined as follows: (1) a forest concatenated with a single node labeled $\square$ is a 1-ary context; (2) the result of concatenating a forest, a tree where the children of the root form a $d$-ary context, and a node labeled $\square$ is a $(d+1)$-ary context.

Now, to a $d$-ary context $K$ we associate a function $\widehat K:Q^d\to Q$. If $d = 1$ and $K = F
\cdot \square$ then $\widehat K = \widehat F$ ; otherwise, let $K = F \cdot c(K_1) \cdot \square$, then $\widehat K(q_1,\ldots,q_d) := \delta(c,\widehat F(\widehat K_1(q_2,\ldots,q_d)),q_1)$.

These functions help us model the information that we store in a run of the algorithm. At any step of a given run, if the depth of the current node is $d$, then the prefix read so far is associated to a $d$-ary context $K_{\sf curr}$. It can be visualized by appending $d$ copies of the string ``$\square{\sf \,\langle \,\rangle\,\rangle}$'' at the end and decoding the resulting forest.

The information that we store at each point of the algorithm is a function $\alpha:Q\to 2^{\{\bot, \top\}}$ that we define as follows
\begin{itemize}
\item $\top \in \alpha(q)$ iff there exist $q_1,\ldots,q_{d-1}$ such that $\widehat K_{\sf curr}(q_1,\ldots,q_{d-1},q) \in \mathsf{Final}$;
\item $\bot \in \alpha(q)$ iff there exist $q_1,\ldots,q_{d-1}$ such that $ \widehat K_{\sf curr}(q_1,\ldots,q_{d-1},q) \not\in \mathsf{Final}$;
\end{itemize} 
along with a function $\beta:Q\to Q$ defined as $\beta = \widehat B$.  In other words, it corresponds to the function associated to the forest that descends from the current node $c$ and is part of the current prefix.

For a given prefix $w$, we will call the functions $\alpha$ and $\beta$ that would be stored after processing $w$ with the algorithm above as the {\em above} function of $w$ and the {\em below} function of $w$.

\begin{lemma}
\label{lem:boolean:earliest}
Let $w$ be a prefix, let $\alpha, \beta$ be its {\em above} and {\em below} functions, and let $c$ be its last open tag:
\begin{itemize}
    \item  $\alpha ( \delta(c, \beta(q), q')) = \{\top\}$ for all $q,q'\in Q$ iff every forest that agrees with $w$ is accepted;
    \item  $\alpha ( \delta(c, \beta(q), q')) = \{\bot\}$ for all $q,q'\in Q$ iff every forest that agrees with $w$ is rejected.
\end{itemize}
\end{lemma}

\begin{proof}
    Suppose $\alpha ( \delta(c, \beta(q), q')) = \{\top,\bot\}$ for some $q,q'\in Q$. We deduce that there exist states that accept and states that reject. We use the fact that the automaton is trimmed to find forests that render those states, and we construct a way of completing $w$ that is accepted and one that is rejected. 

    Let $w = w_0\, c_1\!\mathop{\langle}w_1\,c_2\!\mathop{\langle} \cdots  \,w_{d-1} \, c_d \!\mathop{\langle} w_B$ where $c_d=c$ and every $w_i$ is well-nested. Consider an arbitrary suffix $w_B' \mathop{\rangle} w_d' \mathop{\rangle} \ldots \mathop{\rangle} w_1' \mathop{\rangle}$ that completes the prefix into a forest $F$ that is accepted (resp. rejected). Relying on the notation from the definition of $\alpha$, we can use the states that are reached after reading $w_1',\ldots,w_d'$ for  $q_1,\ldots,q_d$, and conclude that $\top\in \alpha( \delta(c, \beta(q), q')) $ (\,resp. $\bot\in \alpha( \delta(c, \beta(q), q')) $\,).
\end{proof}

It remains to show how to maintain $c$, $\alpha$, and $\beta$. We rely on the following two lemmas.

\begin{lemma}
\label{lem:boolean:down}
Let $w' = w\, c' \langle\, w_d$ and $w'' = w\, c'\langle\, w_d \, c''\langle$, where  $w_d$ is well-nested.
Let  $\alpha'$ and $\beta'$ be the functions for the {\em above} and {\em below} parts corresponding to $w'$, and similarly, let $\alpha''$ and $\beta''$ be the functions for the {\em above} and {\em below} parts corresponding to $w''$. Then, 
\[ 
\alpha''(q) = \bigcup_{q'\in Q} \alpha' ( \delta(c', \beta'(q), q'))
\quad \textrm{and} \quad \beta'' = \mathrm{id}\,.\] 
For the degenerate case of $w'=w_d$ and $w''=w_d\, c''\langle$ with $w_d$ well-nested, the formulas above hold with $\alpha'=\alpha_{\mathsf{root}}$ and $c'=\$$, where
\[\alpha_{\mathsf{root}}(q)= 
\begin{cases} 
\{\top\} & q \textrm{ is final}\,,\\ 
\{\bot\} & q \textrm{ is not final}\,.
\end{cases}
\]   
\end{lemma}
\begin{proof}
Let $A'$ be the $d$-ary context constituting the above part for $w'$ and $B'$ the forest constituting the below part for $w'$. Similarly,  let $A''$ be the $(d+1)$-context constituting the {\em above} part for $w''$ and $B''$ the forest constituting the {\em below} part for $w''$. Then, $B''$ is empty and  $A''$ is obtained from $A'$ by substituting its deepest node labeled $\square$ with $c(B'\cdot \square)\cdot\square$. The formulas for $\alpha''$ and $\beta''$ follow immediately. The argument for the degenerate case is analogous.
\end{proof}

\begin{lemma}
\label{lem:boolean:up}
Let $w' = w\, c'\langle\,w_d$ and $w'' = w\, c'\langle\, w_d\,\rangle$, where  $w_d$ is well-nested. Using the naming convention from Lemma~\ref{lem:boolean:down}, we have 
\[ 
\alpha'' = \alpha 
\quad \textrm{and} \quad \beta'' (q) = \beta(\delta(c',\beta'({\sf init}),q))\,.\]
\end{lemma}

\begin{proof}
Let $A$ be the $d$-ary context constituting the {\em above} part for $w$ and $B$ the forest constituting the {\em below} part for $w$.  Similarly,  let $A'$ be the $(d+1)$-ary context constituting the {\em above} part for $w'$ and $B'$ the forest constituting the {\em below} part for $w'$. Finally, let $A''$ be the $d$-ary context constituting the {\em above} part for $w''$ and $B''$ the forest constituting the {\em below} part for $w''$. Then, $A = A''$ and $B'' = B\cdot c'(B')$.
Again, the formulas for $\alpha''$ and $\beta''$ follow immediately.
\end{proof}

Based on Lemmas~\ref{lem:boolean:down}~and~\ref{lem:boolean:up} we can maintain $c$, $\alpha$, and $\beta$ using a stack. Before going down we push the current $\langle \alpha, \beta, c\rangle$ on the stack and compute new $\alpha$ and $\beta$ using Lemma~\ref{lem:boolean:down}. The new $c$ is read from the stream. Before going up, we pop old $\langle\alpha, \beta,c\rangle$ from the stack, and compute new $\alpha$ and $\beta$ using Lemma~\ref{lem:boolean:up}. Before reading the first tag, we also check if the automaton accepts everything or nothing, and if this is the case we answer accordingly. If not, we process the tags one by one, maintaining $\alpha$, $\beta$, $c$ and return an answer as soon as the criteria in Lemma~\ref{lem:boolean:earliest} becomes satisfied. 
It is evident the updates can be made in constant time and the space taken---size of the stack---is proportional to the tree depth. Thus, \autoref{thm:boolean-case} is proven.

\section{Unary queries}  \label{sec:unary}

We now generalize the described solution to unary queries. We assume that a unary query is represented by an automaton $\mathcal{A} = (\Sigma^{\mathsf{out}}, Q, \delta, {\sf init}, {\sf Final})$ that inputs a forest with a single node marked as the answer node. A node $v$ in forest $F$ is returned by the query iff the automaton accepts the forest obtained from $F$ by marking $v$, as explained in Section~\ref{sec:preliminaries}. 
 
 Compared to the Boolean case, instead of a single triple $\langle\alpha, \beta, \gamma\rangle$, we need to maintain one triple for each possible marking of the part of the forest read so far. 
Let $\langle\alpha_v, \beta_v, \gamma_v\rangle$ be the triple for the marking that distinguishes node $v$, and $\langle\alpha_\emptyset, \beta_\emptyset, \gamma_\emptyset\rangle$ the triple for the marking that distinguishes no nodes (meaning that the node to be returned has not been read yet). Without loss of generality we can assume that both $\alpha_v$ and $\beta_v$ indicate whether the corresponding part of the forest contains the marked node or not: for this, it suffices to assume that the automaton remembers whether a marked node has already been seen; similarly, $\gamma_v$ is either a label $c$ or a marked label $\check c$. Let $\mathsf{Above}$, 
$\mathsf{Below}$, and 
$\mathsf{Current}$ be the sets of triples $\langle\alpha, \beta, \gamma\rangle$ for which the marked node is in the part of the forest corresponding to $\alpha$, $\beta$, and $\gamma$, respectively. Let $\mathsf{Marked} = \mathsf{Above} \,\cup\,  \mathsf{Below} \,\cup\, \mathsf{Current}$ and let $\mathsf{Unmarked}$ be the set of triples corresponding to unmarked prefixes.

As we continue reading the forest, the triples  $\langle\alpha_v, \beta_v, \gamma_v\rangle$ and $\langle\alpha_\emptyset, \beta_\emptyset, \gamma_\emptyset\rangle$ evolve just like in the Boolean case. Moreover, with each opening tag read, we create a new triple, corresponding to the marking that distinguishes the current node, and insert it into our collection. Similarly to the Boolean case, node $v$ can be returned or discarded as soon as $\langle\alpha_v, \beta_v, \gamma_v\rangle$ satisfies the conditions spelled out in Lemma~\ref{lem:boolean:earliest}. In either case, we consider this to be a deletion of $v$ along with $\langle\alpha_v, \beta_v, \gamma_v\rangle$ from our collection. The challenge is to implement this algorithm with constant cost of updates, deletes, and inserts. Importantly, not only the cost of a single operation must be constant (this is for free), but the total cost of all operations to be performed when a tag is read. For this purpose, we organize the tuples into a suitable data structure.

\subsection{Data structure} 
\label{ssec:data-structure}
The data structure stores candidate nodes grouped by their associated triples: at any moment, given a triple $\langle \alpha, \beta, \gamma \rangle$, it provides access to candidates $v$ such that $\langle \alpha_v, \beta_v, \gamma_v\rangle = \langle \alpha, \beta, \gamma\rangle$. The triple $\langle \alpha_\emptyset, \beta_\emptyset, \gamma_\emptyset \rangle$ will be stored outside of the data structure. The data structure must support adding new candidates, regrouping the candidates as their triples evolve, and extracting candidates with indicated triple. In order to support the evolution of triples, we need access to their previous values. 

Instead of presenting the outputs via constant-delay enumeration, we will do it via iterators. This is a subtler notion, as they navigate precisely the set of outputs without performing any other extra operations that might be hidden in the constant delay.

Abstractly, the \emph{multistack} data structure maintains a set $X$ of elements along with a labelling $h: X \to \mathcal{K}^+$ for a fixed alphabet $\mathcal{K} = \mathcal{K}_1 \,\cup\, \mathcal{K}_2$ with $\mathcal{K}_1 \,\cap\, \mathcal{K}_2 = \emptyset$. We think of the word $h(x)$ as a private stack of element $x$, growing to the right. The color of element $x$ is the element at the top of its stack: if $h(x) = k_1k_2\cdots k_n$ then the color of $x$ is $k_n$.  The following constant-time operations are supported:
\begin{itemize}
\item {\sbweight Create} a new empty structure;
\item {\sbweight Insert} an element $x$ of color $k\in \mathcal{K}_1$: add $x$ to $X$ and let $h(x)=k$;
\item {\sbweight Extract} elements of color $k\in \mathcal{K}$: remove them from $X$ and return an iterator giving access to these elements;
\item {\sbweight Multi-push} according to $f: \mathcal{K} \to \mathcal{K}_2$: append $f(k)$ to $h(x)$ for each element of color $k$;
\item {\sbweight Multi-pop}: remove the last symbol from $h(x)$ for each $x$ with $|h(x)| >1$;
\item {\sbweight Multi-swap} according to a function $f: \mathcal{K} \to \mathcal{K}$ such that $f(\mathcal{K}_1)\subseteq \mathcal{K}_1$ and $f(\mathcal{K}_2)\subseteq \mathcal{K}_2$: replace with $f(k)$ the last label of $h(x)$ for each element $x$ of color $k$.
\end{itemize}

Above, $f(\mathcal{K}_i)$ stands for $\{f(k) \bigm| k\in \mathcal{K}_i\}$. Note that the extract operation returns an iterator in constant time, but actually iterating over its elements will take linear time. 

From the definitions of the operations if follows that the  multistack satisfies the following invariant:
\begin{itemize}
\item all elements of colors from $\mathcal{K}_1$ have history of length  1;
\item all elements of colors from $\mathcal{K}_2$ have history of length at least 2.
\end{itemize}
See~\cref{fig:api_ex} for an illustration of the data structure operation.
\begin{figure}
    \centering
    \scalebox{0.9}[0.9]{
    \begin{tikzpicture}[scale=1.0,
  x=0.72cm,
  y=0.70cm,
  every node/.style={font=\scriptsize},
  cell/.style={
    draw=black,
    fill=white,
    minimum width=5mm,
    minimum height=3.8mm,
    inner sep=0pt
  },
  top cell/.style={
    cell,
    fill=black!12,
    line width=0.75pt
  },
  elem/.style={
    anchor=east,
    inner sep=1pt
  },
  op/.style={
    anchor=east,
    align=right,
    font=\scriptsize
  },
  sep/.style={
    draw=black!35,
    line width=0.35pt
  }
]

% left column: insert
\node[op] at (-1.35,-1.0) {$\mathrm{insert}(x_4,k_b)$};
\node[elem] at (-0.65,-0.45) {$x_1$};
\node[cell] at (-0.1,-0.45) {$k_a$};
\node[top cell] at (0.6,-0.45) {$k_b$};
\node[elem] at (-0.65,-1.0) {$x_2$};
\node[top cell] at (-0.1,-1.0) {$k_a$};

\node[elem] at (2.35,-0.45) {$x_1$};
\node[cell] at (2.9,-0.45) {$k_a$};
\node[top cell] at (3.6,-0.45) {$k_b$};
\node[elem] at (2.35,-1.0) {$x_2$};
\node[top cell] at (2.9,-1.0) {$k_a$};
\node[elem] at (2.35,-1.55) {$x_4$};
\node[top cell] at (2.9,-1.55) {$k_b$};

% left column: extract
\node[op] at (-1.35,-3.35) {$\mathrm{extract}(k_b)$};
\node[elem] at (-0.65,-2.8) {$x_1$};
\node[cell] at (-0.1,-2.8) {$k_a$};
\node[top cell] at (0.6,-2.8) {$k_b$};
\node[elem] at (-0.65,-3.35) {$x_2$};
\node[top cell] at (-0.1,-3.35) {$k_b$};
\node[elem] at (-0.65,-3.9) {$x_3$};
\node[cell] at (-0.1,-3.9) {$k_c$};
\node[top cell] at (0.6,-3.9) {$k_a$};

\node[elem] at (2.35,-3.35) {$x_3$};
\node[cell] at (2.9,-3.35) {$k_c$};
\node[top cell] at (3.6,-3.35) {$k_a$};

% left column: multi-push
\node[op] at (-1.35,-5.68)
  {$\mathrm{multi\text{-}push}(f)$\\[-0.25ex]
   $f(k_a)=\ell_a$\\[-0.25ex]
   $f(k_b)=\ell_b$};
\node[elem] at (-0.65,-5.4) {$x_1$};
\node[cell] at (-0.1,-5.4) {$k_a$};
\node[top cell] at (0.6,-5.4) {$k_b$};
\node[elem] at (-0.65,-5.95) {$x_2$};
\node[top cell] at (-0.1,-5.95) {$k_a$};

\node[elem] at (2.35,-5.4) {$x_1$};
\node[cell] at (2.9,-5.4) {$k_a$};
\node[cell] at (3.6,-5.4) {$k_b$};
\node[top cell] at (4.3,-5.4) {$\ell_b$};
\node[elem] at (2.35,-5.95) {$x_2$};
\node[cell] at (2.9,-5.95) {$k_a$};
\node[top cell] at (3.6,-5.95) {$\ell_a$};

% right column: multi-pop
\node[op] at (8.15,-1.0) {$\mathrm{multi\text{-}pop}$};
\node[elem] at (9.1,-0.45) {$x_1$};
\node[cell] at (9.65,-0.45) {$k_a$};
\node[cell] at (10.35,-0.45) {$k_b$};
\node[top cell] at (11.05,-0.45) {$\ell_b$};
\node[elem] at (9.1,-1.0) {$x_2$};
\node[cell] at (9.65,-1.0) {$k_a$};
\node[top cell] at (10.35,-1.0) {$\ell_a$};
\node[elem] at (9.1,-1.55) {$x_3$};
\node[top cell] at (9.65,-1.55) {$k_c$};

\node[elem] at (12.0,-0.45) {$x_1$};
\node[cell] at (12.55,-0.45) {$k_a$};
\node[top cell] at (13.25,-0.45) {$k_b$};
\node[elem] at (12.0,-1.0) {$x_2$};
\node[top cell] at (12.55,-1.0) {$k_a$};
\node[elem] at (12.0,-1.55) {$x_3$};
\node[top cell] at (12.55,-1.55) {$k_c$};

% right column: multi-swap
\node[op] at (8.15,-3.43)
  {$\mathrm{multi\text{-}swap}(g)$\\[-0.25ex]
   $g(k_a)=\ell_a$\\[-0.25ex]
   $g(k_b)=\ell_b$};
\node[elem] at (9.1,-3.15) {$x_1$};
\node[cell] at (9.65,-3.15) {$k_a$};
\node[top cell] at (10.35,-3.15) {$k_b$};
\node[elem] at (9.1,-3.7) {$x_2$};
\node[top cell] at (9.65,-3.7) {$k_a$};

\node[elem] at (12.0,-3.15) {$x_1$};
\node[cell] at (12.55,-3.15) {$k_a$};
\node[top cell] at (13.25,-3.15) {$\ell_b$};
\node[elem] at (12.0,-3.7) {$x_2$};
\node[top cell] at (12.55,-3.7) {$\ell_a$};

\draw[sep] (5.25,0.05) -- (5.25,-6.15);

\end{tikzpicture}
    }
    \caption{Examples of operations}
    \label{fig:api_ex}
\end{figure}
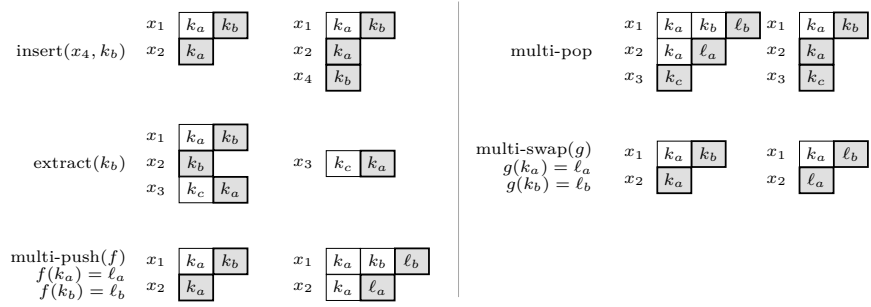

\subsection{Algorithm}

The algorithm uses a stack over the alphabet $\mathsf{Unmarked}$
to maintain the triple $\langle\alpha_\emptyset, \beta_\emptyset, c_\emptyset\rangle$ as in the unary case and a multi-stack over the alphabet  \[ \mathcal{K}=\mathsf{Marked}\,\cup\, \mathsf{Below} \times \{\mathsf{tmp}\}\,,\] 
split into $\mathcal{K}_1 = {\sf Below} \,\cup\, {\sf Current} \,\cup\, {\sf Below} \times \{\mathsf{tmp}\}$ and $\mathcal{K}_2 = \mathsf{Above}$,
to maintain the set of candidate nodes $v$ colored with $\langle\alpha_v, \beta_v, \gamma_v\rangle$. Additionally, after processing each tag  the multistack satisfies the following conditions:
\begin{itemize}
\item there are no elements of colors from  $\mathsf{Below}\times\{\mathsf{tmp}\}$;
\item the length of the history of a candidate node is one plus the distance between the current node and the closest common ancestor of the current node and the candidate node. 
\end{itemize}

After processing each tag, we will perform the  following \emph{emit} and \emph{flush} subroutines:
\begin{itemize}
\item[(\emph{emit})] for each $\langle\alpha,\beta, \gamma\rangle \in \mathsf{Marked}$ such that $\mathrm{Range}(\langle\alpha,\beta, \gamma\rangle) = \{\top\}$ extract elements of color $\langle\alpha,\beta, \gamma\rangle$ and emit the iterator;
\item[(\emph{flush})] for each $\langle\alpha,\beta, \gamma\rangle \in \mathsf{Marked}$ such that $\mathrm{Range}(\langle\alpha,\beta, \gamma\rangle) = \{\bot\}$ extract elements of color $\langle\alpha,\beta, \gamma\rangle$ and discard the iterator;
\end{itemize}
where \[\mathrm{Range}\big(\langle \alpha, \beta, \gamma\rangle\big) = \bigcup_{q,q'\in Q} \alpha \big( \delta(\gamma, \beta(q), q')\big)\,. \] 

When processing opening and closing tags, we use two operations on triples, which are based on Lemmas~\ref{lem:boolean:down} and \ref{lem:boolean:up}. 
\begin{itemize}
    \item For a triple $\langle\alpha, \beta, \gamma\rangle$ and $\gamma'$, we define \[\mathrm{Open}_{\gamma'}\big(\langle\alpha, \beta,\gamma\rangle\big) = \langle \alpha', \beta', \gamma'\rangle\] where $\alpha'(q) = \bigcup_{q'\in Q} \alpha(\delta(\gamma, \beta(q),q'))$ and $\beta'=\mathrm{id}$.
    \item For triples $\langle\alpha, \beta,\gamma\rangle$ and  $\langle\alpha', \beta',\gamma'\rangle$, we define \[\mathrm{Close}\big( \langle\alpha, \beta,\gamma\rangle, \langle\alpha', \beta',\gamma'\rangle\big) = \langle\alpha'', \beta'', \gamma''\rangle\] where $\alpha''=\alpha$, $\beta''(q) = \beta(\delta(\gamma', \beta'({\sf init}),q))$, and $\gamma''=\gamma$.
\end{itemize}

We begin with both the stack and the multi-stack empty.  Before reading any tags, we push the triple $\langle\alpha_{\mathsf{root}}, \mathrm{id},\$\rangle$ to the stack. 
Recall that \[\alpha_{\mathsf{root}}(q)= 
\begin{cases} 
\{\top\} & q \textrm{ is final}\\ 
\{\bot\} & q \textrm{ is not final}
\end{cases}
\]
and $\$$ is a reserved element of $\Sigma$ such that $\delta(\$,q,q')=q$.
The stack is now non-empty and it will always be non-empty after a tag is processed, unless an unbalanced closing tag is encountered.  

When an opening tag $c\,\langle$ is read,
\begin{itemize}
\item let $\langle\alpha_\emptyset, \beta_\emptyset, \gamma_\emptyset\rangle$ be the triple at the top of the stack;
\item push to the stack the triple $\mathrm{Open}_{c}\big(\langle\alpha_\emptyset, \beta_\emptyset, \gamma_\emptyset\rangle\big)$;
\item multi-push according to the function \[\mathsf{Marked}  \ni \langle\alpha, \beta, \gamma\rangle \mapsto \mathrm{Open}_{c}\big(\langle\alpha, \beta, \gamma\rangle\big) \in \mathsf{Above}\] extended by identity to $\mathsf{Below}\times\{\mathsf{tmp}\}$ (which does not matter because there are no elements of colours from $\mathsf{Below}\times\{\mathsf{tmp}\}$ in the multi-stack);
\item insert  the current node  with colour  $\mathrm{Open}_{\check{c}}\big(\langle\alpha_\emptyset, \beta_\emptyset, \gamma_\emptyset\rangle\big) \in \mathsf{Current}$ into the multi-stack;
\item emit, flush, and continue.
\end{itemize}

When a closing tag $\rangle$ is read,  
\begin{itemize}
\item pop $\langle\alpha_\emptyset,\beta_\emptyset, \gamma_\emptyset\rangle$ from the stack;
\item if the stack becomes empty, terminate with an error;
\item pop  $\langle\alpha'_\emptyset,\beta'_\emptyset, \gamma'_\emptyset\rangle$ from the stack;
\item push to the stack the triple $\mathrm{Close}\big(\langle\alpha'_\emptyset,\beta'_\emptyset, \gamma'_\emptyset\rangle, \langle\alpha_\emptyset,\beta_\emptyset, \gamma_\emptyset\rangle\big)$;
\item  multi-swap according to the function 
\[\hspace{-3em}\mathsf{Below} \,\cup\, \mathsf{Current}\ni\!\langle \alpha, \beta, \gamma\rangle
\mapsto \big(\mathrm{Close}\big(
\langle 
    \alpha'_\emptyset, 
    \beta'_\emptyset, 
    \gamma'_\emptyset
\rangle,
\langle 
    \alpha, 
    \beta, 
    \gamma
\rangle 
\big), \mathsf{tmp}\big) \!\in \mathsf{Below} \times \{\mathsf{tmp}\}\]
extended by identity to $\mathsf{Above}\cup\mathsf{Below}\!\times\!\{\mathsf{tmp}\}$;
\item multi-pop;
\item multi-swap according to the function 
\begin{align*}
\mathsf{Marked}
\ni
\langle \alpha', \beta', \gamma'\rangle
&
\mapsto 
\mathrm{Close}\big(
\langle 
    \alpha', 
    \beta', 
    \gamma'
\rangle,
\langle 
    \alpha_\emptyset, 
    \beta_\emptyset, 
    \gamma_\emptyset
\rangle 
\big) \in \mathsf{Marked}\,,\\
\hspace{-2em}\mathsf{Below}\times \{\mathsf{tmp}\} \ni \big(\langle \alpha, \beta, \gamma\rangle,\mathsf{tmp}\big) &
\mapsto
\langle \alpha, \beta, \gamma\rangle \in \mathsf{Below}\,;
\end{align*}
\item emit, flush, and continue.
\end{itemize}

Note that the functions used in multi-swaps respect the split of $\mathcal{K}$ into $\mathcal{K}_1$ and $\mathcal{K}_2$. In the first multi-swap this is obvious. In the second one, the second rule maps elements of $\mathcal{K}_1$ to elements of $\mathcal{K}_1$. To see the constraint holds for the first rule, recall that $\mathrm{Close}\big(
\langle 
    \alpha', 
    \beta', 
    \gamma'
\rangle,
\langle 
    \alpha_\emptyset, 
    \beta_\emptyset, 
    \gamma_\emptyset
\rangle 
\big)=
\langle
\alpha', 
\beta\circ\beta',
\gamma'
\rangle
$
where $\beta(q)=\delta(\beta_\emptyset(\mathsf{init}),q)$. Clearly, $\langle\alpha',\beta',\gamma'\rangle\in\mathsf{Above}$ implies
$
\langle
\alpha', 
\beta\circ\beta',
\gamma'
\rangle \in \mathsf{Above}
$, and similarly for $\mathsf{Below}$ and $\mathsf{Current}$. Hence, the first rule also respects the split into $\mathcal{K}_1$ and $\mathcal{K}_2$.

Correctness of the algorithm follows from Lemmas~\ref{lem:boolean:earliest}--\ref{lem:boolean:up} applied separately to each node, the choice of the initial triple pushed to the ordinary stack,  and the semantics of the multi-stack operations. The complexity guarantees follow from the guarantees for the multi-stack. 

\section{Data structure implementation}  \label{sec:ds}

It remains to show that we can implement the Candidate Data Structure operations to work in constant time.
This is not trivial, since a naïve implementation---storing explicit stacks for each element in the structure---would result in a linear number of operations for each \emph{multi-push}, \emph{multi-pop}, and \emph{multi-swap}.

\subsection{Intuition}

Before defining the implementation let us outline the challenges that arise from the requirements stated in \autoref{ssec:data-structure} and how they drive the design given in \autoref{ssec:definition-of-data-structure}.

\paragraph{Buckets} The key insight is that the effect of these operations is the same for all elements of the same color.
Instead of physically maintaining separate stacks for each element we can maintain \emph{buckets}, one for each color.
In each bucket we hold a \emph{list of elements} that currently inhabit that color.
If we did not have to maintain the color history---i.e. support popping---the implementation would be straightforward, as all other operations amount to adding to buckets, relabeling, or emptying them.

\paragraph{Logs} To maintain the history we keep a \emph{log of events} for each bucket, implemented with a stack.
On the stack we put \emph{pointers} to the heads of the lists that were moved into the given bucket in the previous push.
When we push, we move the heads of the lists into their new buckets, which might result in a concatenation of two or more lists.
When we pop, we take the pointers to those heads, decouple the lists, and move them back to appropriate buckets.

\paragraph{Invalidate} Supporting \emph{extract} introduces nuance.
The naïve implementation---walking through the list in the bucket---has two issues.
Firstly, it can cause\linebreak\emph{staleness}, as log entries might point to nodes that were already emitted;
a \emph{multi-pop} could then reintroduce a node back into the data structure after it was extracted.
This could be alleviated by an additional level of indirection in the node that would contain a staleness bit, so that popping could identify and ignore stale nodes.
The second issue, however, is that such an implementation does not provide constant delay enumeration: each node has to have its staleness bit flipped before all of them are returned.
To obtain our result we cannot touch every \emph{node} that gets extracted, instead we need to maintain and propagate information about on the level of \emph{buckets}.
The idea is to have \emph{extract} operations mark the emptied bucket with an \emph{Invalidate} event on top of its log.
During the pop, when we encounter a bucket with a fresh \emph{Invalidate} event on top of its log, we know that any lists that get moved back \emph{from that bucket} have been already emitted.
Because a push always moves entire buckets, this means that all the source buckets are now stale and we can propagate the \emph{Invalidate} event back to them.
This only works under the assumption that $\mathcal{K}$ is two disjoint sets---$\mathcal{K}_1$ is never stale and always has an empty log.
Otherwise, inserting to a bucket with an \emph{Invalidate} event would be possible and break the data structure.

\subsection{Formal definition}\label{ssec:definition-of-data-structure}

We will now formalize this intuition and provide its implementation in pseudocode (Algorithm~\autoref{fig:data-structure-pseudocode}).
The Candidate Data Structure over $\mathcal{K} = \mathcal{K}_1 \cup \mathcal{K}_2$ holds:

\newcommand{\op}{\mathsf{op}}
\newcommand{\head}{\mathsf{head}}
\newcommand{\tail}{\mathsf{tail}}

\begin{enumerate}
    \item a global \emph{operation counter} $\op$;
    \item a collection of $|\mathcal{K}|$ \emph{buckets}: $(B_k)_{k\in \mathcal{K}}$, each containing either nothing or a list of candidates defined by two \emph{candidate nodes}: head and tail, accessible via functions $\head(\cdot)$ and $\tail(\cdot)$.
    \item for each bucket $B_k$ a \emph{log} $E_k$: a stack of \emph{events}; an event can be of one of two kinds:
    \begin{itemize}
        \item a \emph{Move event}: a pointer to a candidate node, a color, and the timestamp of its source \emph{multi-push} as given by the operation counter; or
        \item an \emph{Invalidation event}: a token with a timestamp corresponding to the source \emph{extract} operation.
    \end{itemize}
\end{enumerate}

\noindent
A \emph{candidate node} is a doubly-linked list node holding a candidate, its direct successor, and direct predecessor.
This is sufficient to implement all our operations.
We assume a fixed, consistent order of iteration over $\mathcal{K}$.

\algrenewcommand\algorithmicindent{0.75em}

\begin{algorithm}[!h]
	\caption{Pseudocode of the data structure.}\label{fig:data-structure-pseudocode}	
	\smallskip
	\begin{algorithmic}[1]
    %\vspace{-3em}
            \footnotesize
		\hspace{-3em}
		\begin{varwidth}[t]{1.0\textwidth}
		\Procedure{Create()}{} 
        \State ${\sf op} \gets 0$
        \For{$k\in\mathcal{K}$}
        \State $B_k \gets \text{None}$
        \State $E_k \gets \textbf{new}\ {\sf Stack}$
        \EndFor
		\EndProcedure
        
        \vspace{0.4em}
        
		\Procedure{Insert}{$x$, $k$}
        \State $N \gets \textbf{new}\ {\sf Node}(x)$
        \If{$B_k = \text{None}$}
        \State $\head(B_k)\gets N$
        \State $\tail(B_k)\gets N$
        \Else
        \State $M \gets \head(B_k)$
        \State $\head(B_k)\gets N$
        \State $N.{\sf next} \gets M$
        \State $M.{\sf prev} \gets N$
        \EndIf
		\EndProcedure

        \vspace{0.4em}

    \Procedure{MultiPop()}{}
      \For{$k\in\mathcal{K}$}
        \State $B^{\sf new}_k \gets B_k$
        \State $L_k \gets {\bf new}\ {\sf List}$
        \State ${\sf invSrc}_k \gets {\bf False}\ , \ {\sf invDest}_k \gets {\bf False}$
      \EndFor
      %\For{$k\in\mathcal{K}_2$}
        % \State {\sf invSrc} $\gets$ array of $|\mathcal{K}|$ {\bf False} values
        % \State {\sf invDest} $\gets$ array of $|\mathcal{K}|$ {\bf False} values
      %\EndFor
      \For{$k\in\mathcal{K}_2$}
        \While{$E_k.{\sf top}.{\sf timestamp} = {\sf op} - 1$}
        \State $e' \gets E_k.{\sf pop}()$
          \If{$e'$ is Invalidate}:
            \State ${\sf invSrc}_k \gets {\bf True}$
          \Else \Comment{$e'$ is Move}
            \State {\bf let} $e' = \text{Move}({\sf op}^*, N^*, k^*)$
            \If{${\sf invSrc}_k$}:
              \If{$k^* \in \mathcal{K}_2$}:
                \State ${\sf invDest}_{k^*} \gets {\bf True}$
            \Else
              \EndIf
              \State $\ell^* \gets (\head:N^*, \tail:\tail(B^{\sf new}_k))$
              \State $L_{k^*}.{\sf append}(\ell^*)$
              \If{$N^*.{\sf prev} \neq \text{None}$}
                \State $\tail(B^{\sf new}_k) \gets N^*.{\sf prev}$
                \State $N^*.{\sf prev}.{\sf next} \gets \text{None}$
                \State $N^*.{\sf prev} \gets \text{None}$
              \Else
                \State $B^{\sf new}_k \gets \text{None}$
              \EndIf
            \EndIf
          \EndIf
        \EndWhile
      \EndFor
      \For{$k\in\mathcal{K}$}
        \For{$\ell^*\in L_k$}
        \State {\sc AssignOrConcat}$(B^{\sf new}_{k},\ell^*)$
          \If{$k \in \mathcal{K}_2\ {\bf and}\ {\sf invDest}_k$}
            \State $E_k.{\sf push}(\text{Invalidate}({\sf op} - 1))$
          \EndIf
        \EndFor
      \EndFor
      \State ${\sf op} \gets {\sf op} - 1$
    \EndProcedure

		\end{varwidth}
        \hspace{2.5em}
		\begin{varwidth}[t]{0.8\textwidth}

    \Procedure{MultiPush}{$f$}:
      \For{$k\in\mathcal{K}$}
        \State $B^{\sf new}_k \gets \text{None}$
      \EndFor
      \For{$k\in\mathcal{K}$}
        \If{$B_k \neq \text{None}$}
          \State {\sc AssignOrConcat}$(B^{\sf new}_{f(k)},B_k)$
          \State $E_{f(k)}.{\sf push}(\text{Move}({\sf op}, \head(B_k), k))$
        \EndIf
      \State $B \gets B^{\sf new}$
      \State ${\sf op} \gets {\sf op} + 1$
      \EndFor
    \EndProcedure
        
        \vspace{0.4em}
        
    \Procedure{MultiSwap}{$f$}
      \For{$k\in\mathcal{K}$}
        \State $B^{\sf new}_k \gets B_k$
        \State $E^{\sf move}_k \gets {\bf new}\ {\sf Stack}$
      \EndFor
      \For{$k\in\mathcal{K}$}
        \If{$B_k \neq \text{None}$}
          \State {\sc AssignOrConcat}$(B^{\sf new}_{f(k)},B_k)$
        \EndIf
        \While{$E_k.{\sf top}.{\sf timestamp} = {\sf op} - 1$}
          \State $E^{\sf move}_{f(k)}.{\sf push}(E_k.{\sf pop}())$
        \EndWhile
      \EndFor
      \For{$k\in\mathcal{K}$}
        \State $B_k \gets B^{\sf new}_k$
        \While{$E^{\sf move}_k.{\sf top} \neq {\sf None}$}
          \State $E_k.{\sf push}(E^{\sf move}_k.{\sf pop}())$
        \EndWhile
      \EndFor
    \EndProcedure
        
        \vspace{0.4em}
        
            \Procedure{extract}{$k$}
      \If{$B_k \neq \text{None}$}
        \State $N \gets\head(B_k)$
        \State $B_k \gets \text{None}$
        \State $N^* \gets {\bf new}\ {\sf Node}(N.{\sf candidate})$
        \State $N^*.{\sf next} \gets N.{\sf next}$
        \State $N.{\sf candidate} \gets \text{None}$
        \State $N.{\sf next} \gets \text{None}$
        \If{$k \in \mathcal{K}_2$}
            \State $E_k.{\sf push}(\text{Invalidate}({\sf op} - 1))$
        \EndIf
        \State {\bf return} $N^*$
      \Else
        \State {\bf return} \text{None}
      \EndIf
    \EndProcedure
    
        \vspace{0.4em}
        
    \Procedure{AssignOrConcat}{$L_1,L_2$}
    \State \Comment{Auxiliary function}
      \If{$L_1 = \text{None}$}
        \State $L_1 \gets L_2$
      \Else
        \State $\tail(L_1).{\sf next} \gets\head(L_2)$
        \State$\head(L_2).{\sf prev} \gets \tail(L_!)$
        \State $\tail(L_1) \gets\tail(L_2)$
      \EndIf
    \EndProcedure
		\end{varwidth}
	\end{algorithmic}
\end{algorithm}

\vspace{0.5em}

\paragraph{\em Create.}
Initialize an empty structure with $|\mathcal{K}|$ empty buckets with empty logs, and set $\op \gets 0$.

\vspace{0.5em}

\paragraph{\emph{Insert} $x$ of color $k \in \mathcal{K}_1$.}
Create a candidate node $N$ holding the candidate $x$.
If $B_k$ is empty, set $N$ as the head and tail of $B_k$;
otherwise, append the current head of $B_k$ after $N$ and set $N$ as the head of $B_k$.

\vspace{0.5em}

\paragraph{\emph{Extract} elements of color $k$.}
Store $\head(B_k)$ in a node $N$ and set $B_k$ to empty.
Return an exact clone of $N$ as the iterator.
To ensure that we do not hold onto the memory of the returned list---recall that $N$ may be referenced in an existing $Move$ event---we nullify $N$'s value and successor.
Finally, if $k \in \mathcal{K}_2$, put an \emph{Invalidate} event on top of $E_k$ with timestamp of $\op$.

\vspace{0.5em}

\paragraph{\emph{Multi-push} according to $f: \mathcal{K} \to \mathcal{K}_2$.}
Let us call the current buckets $(B^{\sf curr}_k)_{k\in\mathcal{K}}$.
We initialize empty buckets $(B^{\sf new}_k)_{k\in\mathcal{K}}$.
For each $k \in \mathcal{K}$:

\begin{itemize}
    \item if $B^{\sf curr}_k$ is empty, there is nothing to be done; otherwise, let $L$ be the list stored in $B^{\sf curr}_k$;
    \item if $B^{\sf new}_{f(k)}$ is still empty, then set it to contain $L$; otherwise, append $L$ after the current list in $B^{\sf new}_{f(k)}$ and set $\tail(B^{\sf new}_{f(k)}) \gets \tail(B^{\sf curr}_k)$;
    \item finally, add an entry pointing to $\head(B^{\sf curr}_k)$ with timestamp $\op$ and color $k$ to the log $E_{f(k)}$. 
\end{itemize}

\noindent
At the end of the operation we set $B_k \gets B^{\sf new}_{k}$ for each $k$, discard $B^{\sf new}_k$, and set $\op \gets \op + 1$.

\vspace{0.5em}

\paragraph{\em Multi-pop.}
Similarly, let us call the current buckets $(B^{\sf curr}_k)_{k\in\mathcal{K}}$.
We initialize a collection of buckets $(B^{\sf new}_k)_{k\in\mathcal{K}}$ to the same state as the current buckets.
To take into account node staleness we first investigate the tops of all logs $E_k$ for $k \in \mathcal{K}_2$: if the top of $E_k$ is an \emph{Invalidate} event with timestamp equal to $\op$, we remember that the bucket is stale and pop that event.
For each color in $\mathcal{K}_2$ we also remember whether it should be marked as stale after the pop.
Then, the operation proceeds in two phases: first, we detach all nodes that need to be moved from their current lists; after that we perform the move.
We allocate $|\mathcal{K}|$ lists of nodes: phase one will populate them, such that the $k$-th list will contain nodes that need to be moved to bucket $k \in \mathcal{K}$; phase two will empty the lists and perform the merges.
For each $k \in \mathcal{K}_2$ remove each \emph{Move} entry in the log $E_k$ with a timestamp equal to $\op - 1$ and process it:

\begin{itemize}
    \item let $N$ be the node pointed to by the event and let $\ell$ be its color;
    \item if $k$ was detected as stale we can ignore the node as it is guaranteed to have been extracted; however, if $\ell \in \mathcal{K}_2$, we need to mark $\ell$ to become stale after the pop;
    \item otherwise (if $k$ is not stale), put $N$ into the $\ell$-th node list, detach it from its preceding node (if any), and set $\tail(B^{\sf new}_k)$ to that preceding node (or set $B^{\sf new} \gets {\sf None}$ if no such node exists).
\end{itemize}

After phase one we move the lists from $B^{\sf new}_k$ to $B^{\sf curr}_k$ for each $k \in \mathcal{K}$ and progressively merge each node from the node lists to the end of its bucket (like in multi-push).
Finally, for each color that was marked as becoming stale, we push an \emph{Invalidate} event with the timestamp of $\op - 2$. At the end of the operation we discard the copies of buckets and set $\op \leftarrow \op - 1$.

\vspace{0.5em}

\paragraph{\emph{Multi-swap} according to $f: \mathcal{K} \rightarrow \mathcal{K}$.}
This operation is equivalent to the multi-push except for the changes to the log.
After copying all the buckets we also pop the frames of all logs with timestamp equal to $\op - 1$ and put them to the side as a list $(e'_{k})_{k\in\mathcal{K}}$.
After the bucket manipulation is done we move each $e'_k$ to the top of $E_{f(k)}$.

\subsection{Proof of correctness}

\newcommand{\cds}{\mathit{cds}}
\newcommand{\ex}{\mathit{extracted}}

The correctness of the data structure hinges on the correctness of the \emph{multi-pop}.
Indeed, it is easy to see that a sequence of operations without any pops would correctly model the interface from \autoref{ssec:data-structure}.
The only delicate part is the invalidation mechanism.

An important observation is that the only nodes that are relevant are those which were heads of buckets during any \emph{multi-push}.
All other nodes are simply successors of these (perhaps indirectly) and are therefore always moved and extracted the same way as their heads.
We want the data structure to maintain the following invariant:

\begin{lemma}\label{lem:data-structure-invalidation-invariant}
    If in some $E_k$ there is a \emph{Move} log entry with timestamp $\op - 1$ and a head node that has already been extracted, then there is also an \emph{Invalidate} entry with timestamp $\op - 1$ at the top of $E_k$.
\end{lemma}
\begin{proof}
    This is trivially true when there are no \emph{extract} operations.
    Assume an extract on bucket $B_k$ happens.
    The head of $B_k$ is emitted and \emph{Invalidate} with timestamp $\op - 1$ is created on top of $E_k$.
    Assume there exists a \emph{Move} log entry with timestamp $\op - 1$, color $\ell$, and which points to a node $N$ extracted from $B_k$.
    This entry has to reside in $E_k$, since $N$ is in $B_k$.
    Naturally, the condition is satisfied for now.
    Since $N$ is now absent from the data structure after the \emph{extract}, the only other \emph{Move} logs pointing to $N$ have to have a lower timestamp.
    In particular, if there exist any such entries, then one with timestamp $\op - 2$ exists in $E_\ell$.
    When a \emph{multi-pop} happens, the \emph{Invalidate} tag is popped, $B_k$ is marked as stale, and the marker is propagated to $E_\ell$ with the timestamp $\op - 2$, hence the invariant is maintained.
\end{proof}

As long as \autoref{lem:data-structure-invalidation-invariant} holds, the invalidation mechanism succeeds in ensuring
that nodes emitted via \emph{extract} stay out of the buckets and any stale log entries pointing to them are ignored.

\vspace{0.5em}

\noindent\textit{Complexity analysis}.
We observe that {\sc Insert} and {\sc Extract} take constant time, and {\sc Create} and {\sc MultiPush} take $O(|\mathcal{K}|)$ time. The cases of {\sc MultiPop} and {\sc MultiSwap} are more subtle given that they perform a while-loop inside of a for-loop: we observe that (1) each {\sc MultiPush} adds at most $|\mathcal{K}|$ elements with the same timestamp {\sf op} at a time, (2) {\sc Extract} is called at most $|\mathcal{K}|$ times for the same timestamp {\sf op} at a time, and (3) we can inductively argue that {\sc MultiPop} also adds at most $|\mathcal{K}|$ Invalidate events with the same timestamp {\sf op} at a time; and so, the while-loops in  {\sc MultiPop} and {\sc MultiSwap} are entered at most $|\mathcal{K}|$ entered across all $k\in K$ values in the for-loop. We conclude that these procedures also take $O(|\mathcal{K}|)$ time.

The space taken by the data structure is limited by the number of active candidates and at most $|\mathcal{K}|$ entries on logs per each \emph{multi-push} that was not yet popped, which satisfied the requirements of \autoref{ssec:data-structure}.

\section{Related work} \label{sec:related}
Evaluation of queries over streamed trees has been studied extensively, especially for XML and XPath~\cite{gou_efficient_2007, corentin_stackless, munoz_streaming_2024, xml_stream_survey_2013, xml_streams_automata_2004}.
%\martin{more}
This line of work shares our motivation of evaluating queries without materializing the entire input tree, but the works focus on specific query languages or fragments and do not provide an earliest-answering guarantee for arbitrary unary MSO queries.

Earliest query answering was studied in works of Gauwin, Niehren, and Tison~\cite{dediu_bounded_2009,gauwin_queries_2011}, where the central notions are {\em bounded delay} and {\em bounded concurrency}. 
These notions are in spirit close to our work, as they measure logically-minimal resource requirements: the former measures how long an answer has to remain pending, and the latter how many of these answers have to coexist.
The main difference  is that they focus on restrictions of languages where these measures can actually be bounded, whereas our results hold for arbitrary MSO queries.

Arguably, the closest reference to our work is a recent paper of Muñoz and Riveros~\cite{munoz_streaming_2024}, which studies streaming enumeration over nested documents for MSO queries. Their framework provides strong enumeration guarantees, including constant update time in data complexity and output-linear delay enumeration (which implies constant delay for unary queries). 
However, they do not guarantee earliest query answering.
They discuss how to perform a form of early answering via $\Delta$-enumeration, and this allows the algorithm to print an output only if could not have been printed at an earlier step. 
However, earliest query answering has an additional condition: an output should be printed when it is guaranteed to be a solution for {\em any} continuation of the tree. 
As far as we know, there is no way to modify a visibly pushdown annotator $\mathcal{A}$ into $\mathcal{A}'$ such that performing $\Delta$-enumeration on $\mathcal{A}'$ simulates earliest query answering for $\mathcal{A}$.

\section{Conclusions} \label{sec:concl}

We have studied earliest query answering for unary MSO queries over streamed
trees, under data complexity. The query, or equivalently the automaton, is fixed,
and the algorithm processes the input with constant update time while outputting and discarding
answers as soon as the current prefix forces them. Several questions remain
open.

The first question is whether the construction is tight in combined complexity. In
this paper, all constants may depend on the automaton. This is unavoidable to some
extent, but the precise dependence on the query is not well understood. In
particular, compiling MSO formulas into deterministic automata may incur a
large blow-up, and our construction further manipulates objects derived from the
state space of this automaton. It would be interesting to determine whether this
cost is inherent for earliest query answering, or whether substantially smaller
representations are possible for natural query languages such as XPath or
JSONPath fragments.

A second direction concerns the representation of outputs. We have used an
enumeration model in which answers are returned as nodes, via an
iterator. This is natural when the output is viewed as a set of positions, but it
does not exploit situations where the set of answers has a more concise
description. Already over words, one may have queries whose answer set consists,
for example, of all positions between the first two occurrences of a given letter.
In such cases, enumerating every answer may be less appropriate than returning a
symbolic representation of the whole set. Understanding which concise
representations are compatible with earliest query answering seems nontrivial:
one would need to define when such a representation is itself forced by the
current prefix, and how it can be maintained without hiding an implicit
enumeration cost.

Another practical issue is the interaction between query optimization and parsing
optimization. Our model treats the input as a stream of already tokenized 
events, and charges the algorithm once for each symbol that is read. This
abstracts away an important source of efficiency in practical JSONPath engines:
systems such as rsonpath~\cite{gienieczko_supporting_2024} and
JSONSki~\cite{jiang_jsonski_2022} do not merely optimize the logical evaluation
of a query, but also use the query to guide the lexer. They may scan for relevant
structural characters, avoid materialixing irrelevant subtrees, or skip regions of
the input that cannot contribute to the answer. From this perspective, skipping
is not just an implementation trick, but part of the query optimization problem
itself.

It would be interesting to develop a model of earliest query answering in which
the query optimizer is allowed to choose a parsing strategy together with the
streaming evaluation strategy. Such a model should make explicit which fragments
of the input may be skipped, what information about them must still be computed,
and how the cost of these parsing operations is measured. The earliest setting
adds a further difficulty: skipping a fragment is sound only if doing so does not
delay an answer that would already be forced by a prefix inside the skipped
fragment, and does not discard information needed to reject or maintain pending
candidates. Understanding when parsing-level optimizations preserve earliest
answers, and how they can be combined with candidate-optimal memory management,
seems to be an important step toward applying the present theory to practical
JSONPath engines.

Finally, JSONPath filters allow comparisons between selected values, and in
particular may compare subtrees for equality. Such tests go beyond the
finite-state view used in this paper. In a streaming setting, subtree equality is
especially problematic: deciding it may require retaining one subtree while the
other is being read, or comparing two subtrees whose relevant parts have not yet
both appeared in the stream. This interacts badly with both earliest output and
space optimality. A natural open problem is to identify fragments with subtree
equality that still admit earliest evaluation with constant update time, or to
prove that such comparisons necessarily require larger memory, delayed output, or
a different output model. This question is related to the automata-theoretic
approach to efficient XPath evaluation, in particular the linear-time evaluation
result of Bojańczyk and Parys~\cite{BojXPATH}, but the combination of subtree equality with earliest
streaming evaluation appears to require different techniques.
%\input{old_contents}
%\input{contents}

%%
%% The acknowledgments section is defined using the "acks" environment
%% (and NOT an unnumbered section). This ensures the proper
%% identification of the section in the article metadata, and the
%% consistent spelling of the heading.

%%
%% The next two lines define the bibliography style to be used, and
%% the bibliography file.
\bibliographystyle{abbrv}
\bibliography{Earliest}
% Biblio entry: Efficient algorithms for evaluating xpath over streams

%\clearpage
%\appendix
%\input{appendix.tex}

\end{document}